\documentclass[%
 aip,
 amsmath,amssymb,
reprint, onecolumn 
]{revtex4-1}

\usepackage{graphicx}
\usepackage{dcolumn}
\usepackage{bm}

\usepackage[utf8]{inputenc}
\usepackage[T1]{fontenc}
\usepackage{mathptmx}
\usepackage{etoolbox}
\usepackage{xfrac}
\usepackage{amsfonts}
\usepackage{float}
\usepackage{dsfont}
\usepackage{amsmath}
\usepackage{amssymb}
\usepackage{amsthm}
\usepackage{mathtools}
\usepackage{graphicx}
\usepackage{color}
\usepackage{natbib}
\usepackage{changes}
\usepackage{bm}
\usepackage{bbm}
\usepackage[hidelinks]{hyperref}
\hypersetup{
     colorlinks   = true,
     citecolor    = blue,
     linkcolor    = blue
}
\makeatletter
\def\@email#1#2{%
 \endgroup
 \patchcmd{\titleblock@produce}
  {\frontmatter@RRAPformat}
  {\frontmatter@RRAPformat{\produce@RRAP{*#1\href{mailto:#2}{#2}}}\frontmatter@RRAPformat}
  {}{}
}%
\makeatother
\newcommand{\ignore}[1]{}
\newcommand{\nobibentry}[1]{{\let\nocite\ignore\bibentry{#1}}}

\newcommand{\ket}[1]{\left\vert#1\right\rangle}
\newcommand{\bra}[1]{\left\langle#1\right\vert}

\def\ketbra#1#2{{\vert#1\rangle\!\langle#2\vert}}

\DeclareMathAlphabet{\mathcal}{OMS}{cmsy}{m}{n}

\theoremstyle{definition}

\theoremstyle{plain}

\theoremstyle{plain}

\usepackage{graphicx}
\usepackage{braket}
\usepackage{dsfont}
\usepackage[english]{babel}
\usepackage{amssymb, amsmath}
\usepackage{amsthm}
\usepackage{xcolor}
\newtheorem{definition}{Definition}[section]
\newtheorem{proposition}{Proposition}[section]
\newtheorem{lemma}{Lemma}[section]
\newtheorem{theorem}{Theorem}[section]
\newtheorem{observation}{Observation}[section]
\newtheorem{corollary}{Corollary}[section]
\newtheorem{conjecture}{Conjecture}[section]

\newcommand{\be}{\begin{equation}} 
\newcommand{\bes}{\begin{equation*}} 			
\newcommand{\ee}{\end{equation}}
\newcommand{\ees}{\end{equation*}}

\newcommand{\bematrix}{\left(\begin{matrix}}
\newcommand{\ematrix}{\end{matrix}\right)}
\begin{document}

\setcitestyle{numbers}
\setcitestyle{square}
\preprint{AIP/123-QED}

\title[Multipartite entanglement in the diagonal symmetric subspace]{Multipartite entanglement in the diagonal symmetric subspace}

\author{Jordi Romero-Pallej\`{a}}
 \affiliation{F\'isica Te\`{o}rica: Informaci\'o i Fen\`{o}mens Qu\`{a}ntics.  Departament de F\'isica, Universitat Aut\`{o}noma de Barcelona, 08193 Bellaterra, Spain}
 \email{jordi.romero.palleja@uab.cat}

\author{Jennifer Ahiable}
 \affiliation{F\'isica Te\`{o}rica: Informaci\'o i Fen\`{o}mens Qu\`{a}ntics.  Departament de F\'isica, Universitat Aut\`{o}noma de Barcelona, 08193 Bellaterra, Spain}

\author{Alessandro Romancino}
 \affiliation{Dipartimento di Fisica e Chimica ``E. Segrè'', Università degli Studi di Palermo, Via Archirafi 36, 90123 Palermo, Italy}

\author{Carlo Marconi}
 \affiliation{Istituto Nazionale di Ottica (CNR-INO), Largo Enrico Fermi 6, 50125 Firenze, Italy}
 
\author{Anna Sanpera}
 \affiliation{F\'isica Te\`{o}rica: Informaci\'o i Fen\`{o}mens Qu\`{a}ntics.  Departament de F\'isica, Universitat Aut\`{o}noma de Barcelona, 08193 Bellaterra, Spain}
\affiliation{ICREA, Pg. Llu\'is Companys 23, 08010 Barcelona, Spain}

\date{\today}

\begin{abstract}
We investigate the entanglement properties in the symmetric subspace of $N$-partite $d$-dimensional systems (qudits). As it happens already for bipartite diagonal symmetric states, also in the multipartite case the local dimension $d$ plays a crucial role. Here, we demonstrate that there is no bound entanglement for $d = 3,4 $ and $N = 3$. Using different techniques, we present strong analytical evidence that no bound entanglement exist for any $N$ if $d\leq 4$. Interestingly, bound entanglement of diagonal symmetric states exist for any number of parties, $N\geq 2$, and local dimensions $d\geq 5$.
\end{abstract}

\maketitle

\section{\label{sec:level1}Introduction}
\noindent Symmetry and entanglement emerge as foundational concepts that profoundly shape our understanding of Nature. While symmetry reflects the invariance of physical systems under transformations, entanglement defines the role of genuine quantum correlations in composite systems. Here we unify these notions to reveal the role of symmetry in reducing the complexity of entanglement characterisation. 

Over the years, entanglement has been exploited as a resource for many information tasks which would, otherwise, be unachievable (see, e.g., \onlinecite{horodecki2009quantum} and references therein,  \onlinecite{jozsa2003role,yin2017satellite,epping2017multi}). 
Consequently, the ability to access this resource in quantum states is a dire need for both theorists and experimentalists. Significant advances have been made in the characterisation and quantification of entanglement, despite the considerable challenges it presents. In fact, the separability problem, that is, determining whether a given quantum state is separable or not, is recognised as a difficult task, falling, in general, under the category of NP-hard problems \onlinecite{gurvits2003classical}. 
Nevertheless, several methods have been developed to detect entanglement in specific quantum states, particularly in the case of bipartite systems. So far, the most powerful criterion for separability is based on the positivity of the partial transpose (PPT) of a quantum state \onlinecite{peres1996separability}. The so called PPT-criterion becomes necessary and sufficient when the composite Hilbert space, $\mathcal{H}_{AB}$, has dimension $\mbox{dim}(\mathcal{H}_{AB})\leq 6$ \onlinecite{Horodecki_1996}. While states with a negative partial transposition are always entangled, the converse is not true: such states are known as PPT-entangled states (PPTES). The complete characterization of PPTES is probably unfeasible, but uncovering their structure, in some large families, especially when dealing with multipartite systems (whose larger Hilbert space inevitably hinders the possibility of a simple characterisation) has consequences for understanding  has an intrinsic importance whose larger Hilbert space inevitably hinders the possibility of a simple characterisation, 

Symmetric states, particularly in the bipartite scenario, have been previously addressed in \onlinecite{toth2009entanglement, PhysRevA.99.032312, PhysRevA.104.016401}, even though much work has also been directed towards significantly larger systems \onlinecite{PhysRevA.67.022112}. In this work, we investigate multipartite diagonal symmetric states (DS), which correspond to mixtures of projectors onto symmetric pure states. While for bipartite DS states it is known under which conditions being PPT is necessary and sufficient to ensure separability \onlinecite{wolfe2014certifying,yu2016separability,tura2018separability}, this is not the case in the multipartite scenario, where PPT has to be checked with respect to all partitions and no sufficient conditions to ensure separability are known yet.  

Before proceeding further, let us present our main findings: first, we demonstrate that for DS states with $N=3$, PPT ensures separability for $d=3,4$, while for $d>4$, we show that this is no longer the case. Second, we analyse multipartite DS states for a generic number of parties, $N$, and disclose the structure of their partial transpositions. Further, for $N=2n$ ($n>1$),  we map DS states in $\mathcal{S} ((\mathbb{C}^d)^{\otimes N})$ to \textit{bipartite} symmetric states in $\mathcal{S}(\mathbb{C}^m\otimes \mathbb{C}^m)$, with $m>d$.
With the aid of this technique, we infer the entanglement properties of DS states from their corresponding bipartite symmetric states, thus providing new insights on the structure of entanglement in multipartite systems.\\

Our manuscript is organised as follows: in Section \ref{sec:2}, we introduce the main concepts and tools needed to address our problem and review previous results on which our ideas are built. In Section \ref{sec:3}, we focus on entanglement in the multipartite scenario for arbitrary dimension, and we discuss the consequences that emerge from imposing PPT conditions for $N=3$. Then, we move to the generic case of $N$ parties and demonstrate that imposing PPT leads to a set of conditions that do not allow to ensure separability. In Section \ref{sec:4}, we present an embedding that allows a significant simplification of the separability problem and illustrate our findings with an example. Remarkably, in Section \ref{sec:5}, we introduce a similar technique, based on DS states of qubits, which allows to the study symmetric bipartite states. Finally, we summarise our findings and present some open questions in the conclusions. 

\section{\label{sec:2}Mathematical preliminaries}

\subsection{The symmetric subspace}

\noindent Let $\mathcal{H}=\mathcal{H}_1\otimes \cdots \otimes \mathcal{H}_N$ be a finite dimensional $N$-partite Hilbert space, where $\mathcal{H}_i=\mathbb C^d$ is the Hilbert space associated to the $i$-th party. The symmetric subspace, $\mathcal {S}(\mathcal{H}) \subset\mathcal{H}$, corresponds to the convex set formed by the normalized pure states, $\ket{\Psi_S}\in\mathcal{H}$, that remain invariant under any permutation of the parties. In an abuse of language, we denote by symmetric quantum states, $\rho_S\in \mathcal{B}(\mathcal {S}(\mathcal{H}))$, the convex hull of projectors onto pure symmetric normalised states, i.e.,
\begin{equation}
\rho_{S} =\sum_{k} p^{(k)}_{S} \ketbra{\Psi^{(k)}_S}{\Psi^{(k)}_S}~,
\end{equation} 
with $p^{(k)}_{S} \geq 0~, \sum_{k} p^{(k)}_{S} =1$. Thus, any $\rho_S\in \mathcal{B}(\mathcal {S}(\mathcal{H}))$ is a positive semidefinite operator ($\rho_S\succeq 0$) with unit trace ($\mathrm{Tr} (\rho_S) = 1$), fulfilling the condition
$\Pi_{S} \rho_{S} \Pi_{S} = \Pi_{S} \rho_{S}=\rho_{S} \Pi_{S}=\rho_{S}$, where $\Pi_S$ is the projector onto the symmetric subspace. The dimension of this subspace is $\binom{N+d-1}{d-1}$, which is significantly smaller than the dimension of the Hilbert space $\mathcal{H}$, i.e., $\mbox{dim} (\mathcal H)=d^N$. Let us recall the following result about separability of symmetric states \onlinecite{ichikawa2008exchange}:
\begin{observation}
    Let $\mathcal{H}= (\mathbb{C}^{d})^{\otimes N}$ be a multipartite Hilbert space. A symmetric state $\rho_S\in \mathcal{B}(\mathcal {S}(\mathcal{H}))$ is separable if it can be written as a convex combination of projectors onto symmetric product states, i.e.,
    \begin{equation}
        \rho_{S}=\sum_i p_i \ketbra{e_i e_i \cdots e_i}{e_i e_i\cdots e_i}.
    \end{equation}
\end{observation}

\noindent A convenient basis for the symmetric subspace is given by the Dicke states which correspond to equally weighted superpositions of states with the same number of excitations. As an example, for a system of two qubits, the corresponding Dicke basis reads $\{\ket{00},\frac{1}{\sqrt{2}}(\ket{01}+\ket{10}),\ket{11}\}$, where it is evident that such states form an orthonormal basis. In the general case of multipartite states of local dimension $d$, Dicke states are defined as follows:

\begin{definition}
    The Dicke state $\ket{D_{\mathbf{k}}}$ corresponds to superpositions of  $k_0$ qudits in the state $\ket{0}$, $k_1$ qudits in the state $\ket{1}$, etc, and is expressed as 
    \begin{equation}
        \ket{D_{\mathbf{k}}}=\mathcal{C}(N,\mathbf{k})^{-1/2} \sum_{\pi \in \mathcal{G}_N}\pi(\ket{0}^{\otimes k_0}\;\otimes\cdots \otimes\;\ket{d-1}^{\otimes k_{d-1}}),
    \end{equation}
   where $\mathcal{G}_{N}$ is the group of permutations of $N$ elements, $\pi$ is a permutation operator, $\mathbf{k}=(k_0,k_1,\cdots,k_{d-1})$ is a partition of $N$, i.e., $k_i \geq 0, \; \sum^{d-1}_{i=0} k_i = N$, and  $\mathcal{C}(N,\mathbf{k})$ is a normalisation constant given by
    \begin{equation}
        \mathcal{C}(N,\mathbf{k})=\binom{N}{\mathbf{k}}=\frac{N!}{k_0! k_1! \cdots k_{d-1}!}.
    \end{equation}
\end{definition}

\noindent Thus, a symmetric quantum state can be compactly 
expressed in the Dicke basis as follows: 
\begin{definition}
 \label{def:multiSS}
Any symmetric state,  $\rho_{S}\in \mathcal{B}(\mathcal{S}((\mathbb{C}^d)^{\otimes N}))$, can be written as 
\begin{equation}
\rho_{S} = \sum_{\mathbf{k}\mathbf{k'}}
\alpha_{\mathbf{k}}^{\mathbf{k'}}
\ketbra{D_{\mathbf{k}}}{D_{\mathbf{k'}}} +\emph{h.c.}
\end{equation}
with $ \alpha_{\mathbf{k}}^{\mathbf{k'}} \in \mathbb{C}$, where $\mathbf{k},\mathbf{k'} $ are partitions of $N$.
\end{definition}

\noindent Convex mixtures of projectors onto Dicke states are denoted as diagonal symmetric (DS), since they are diagonal in the Dicke basis. They form a convex subset of $\mathcal{S}$ and are particularly relevant for our analysis.
\begin{definition}
 \label{def:multiDSS}
Any DS state, $\rho_{DS} \in \mathcal{B}(\mathcal{S} ((\mathbb{C}^d)^{\otimes N}))$, is of  the form
 \begin{equation}
  \rho_{DS} = \sum_{\mathbf{k}} p_{\mathbf{k}}\ketbra{D_{\mathbf{k}}}{D_{\mathbf{k}}},
\end{equation}
where $\mathbf{k}$ is a partition of $N$ and $p_{\mathbf{k}} \geq 0, \; \forall \;\mathbf{k}$, $\sum_{\mathbf{k}} p_{\mathbf{k}}=1$.
\end{definition}

\noindent Although in the following we will focus mainly on the multipartite scenario, we briefly recall some basic definitions regarding the symmetric bipartite space of two qudits, i.e., $\mathcal S(\mathbb{C}^{d}\otimes \mathbb{C}^{d})$.  
In this case, the Dicke states take the simpler expression,
\begin{equation}
\label{dick2}
\ket{D^{(d)}_{ii}} = \ket{ii}~, \quad \ket{D^{(d)}_{ij}} = \frac{\ket{ij} + \ket{ji}}{\sqrt{2}}~, \quad   i \neq j~,
\end{equation}
where the superscript reminds that $\{\ket{i}\}_{i=0}^{d-1}$ is an orthonormal basis of $\mathbb{C}^{d}$. Notice that the dimension of $\mathcal{S} (\mathbb{C}^d\otimes \mathbb{C}^d)$ is now  $d(d+1)/2$. 

\noindent Equivalently, bipartite symmetric quantum states can be compactly 
expressed as follows: 
\begin{definition}
 \label{def:biSS}
Any bipartite symmetric state,  $\rho_{S}\in \mathcal{B}(\mathcal{S}(\mathbb{C}^d\otimes \mathbb{C}^d))$, can be written as 
\begin{equation}
\rho_{S} = \sum_{\substack{0\leq i\leq j<d\\
                  0\leq k\leq l<d}}
               \left(\rho_{ij}^{kl} \ketbra{D_{ij}^{(d)}}{D_{kl}^{(d)}} + \emph{h.c.} \right)~,
\end{equation}
with $ \rho_{ij}^{kl} \in \mathbb{C}$. Notice that, due to the symmetry of the Dicke states, it holds that $ \rho_{ij}^{kl} = \rho_{ji}^{kl}= \rho_{ij}^{lk}= \rho_{ji}^{lk} ~\forall \;i,j,k,l$.
\end{definition}

\begin{definition}
 \label{def:biDS}
Any DS state, $\rho_{DS} \in \mathcal{B}(\mathcal{S}(\mathbb{C}^d\otimes \mathbb{C}^d))$, is of  the form
\begin{equation}
  \rho_{DS} = \sum_{0 \leq i \leq j < d} p_{ij}\ketbra{D^{(d)}_{ij}}{D^{(d)}_{ij}},
\end{equation}
with $p_{ij} \geq 0,\; \forall\; i,j$ and $\sum_{ij} p_{ij}=1$.
\end{definition}

\noindent Due to their intrinsic structure, symmetric states possess a natural decomposition:
\begin{definition}
 \label{decsym}
 Every symmetric state, $\rho_{S}$, can be decomposed as $\rho_{S} = \rho_{DS} + \sigma_{CS}$, where $\rho_{DS}$ is a DS state and $\sigma_{CS}$ is a traceless symmetric contribution which contains the coherences between Dicke states. In the bipartite case, this decomposition reads
\begin{equation}
 \label{sym}
 \rho_{S} = \rho_{DS}+\sigma_{CS}
 =\sum_{0\le i \le j < d} p_{ij} \ketbra{D^{(d)}_{ij}}{D^{(d)}_{ij}} + \underset{{(i,j)\neq(k,l)}}\sum\left( {\alpha_{ij}^{kl}} \ketbra{D^{(d)}_{ij}}{D^{(d)}_{kl}} + \emph{h.c.} \right),
\end{equation}
\noindent with  ${\alpha_{ij}^{kl}} \in \mathbb{C}$ and ${\alpha_{ij}^{kl}} = ({\alpha_{kl}^{ij}})^*$. 

\end{definition}

\subsection{Entanglement in bipartite DS states}
    
\noindent Previous works have shown that for bipartite DS states of dimensions $d=3$, and $d=4$, separability is equivalent to PPT  \onlinecite{yu2016separability,tura2018separability}. Below, we review the main steps leading to such a result, as they will become pivotal in the  multipartite case. The crucial point is to realize that despite the partial transpose of a symmetric state in $ \mathcal{S(}\mathbb{C}^d\otimes \mathbb{C}^d$ is not anymore symmetric and spans the full  Hilbert space ($d^2\times d^2$), for diagonal symmetric states the conditions for PPT are encoded in a reduced matrix of size $d\times d$.

\begin{proposition}
    \label{mdef} 
 Let $\rho_{DS} \in  \mathcal{B}(\mathcal{S(}\mathbb{C}^d\otimes \mathbb{C}^d))$. Then, its partial transpose in the computational basis reads
 \begin{equation}
        \rho_{DS}^{\Gamma}=M_d(\rho_{DS}) \bigoplus_{\substack{{0\leq i \neq j <d}}}\frac{p_{ij}}{2}~,
        \label{eqmdef}
    \end{equation}
    where $M_d(\rho_{DS})$ is a $d \times d$ matrix with non-negative entries in the subspace spanned by  $\{\ket{00},\ket{11}\dots,\ket{dd}\}$ 
    \begin{equation}
        M_d(\rho_{DS})=
         \begin{pmatrix}
        \bar{p}_{00} & \bar{p}_{01} & \cdots&\bar{p}_{0 d-1}\\
        \bar{p}_{01} & \bar{p}_{11} & \cdots&\bar{p}_{1 d-1} \\
        \vdots&\vdots&\ddots&\vdots\\
         \bar{p}_{0 d-1} & \bar{p}_{1 d-1} & \cdots& \bar{p}_{d-1 d-1}
    \end{pmatrix} 
    \label{mbipartit}
    \end{equation}
    and $\bar{p}_{ij}=\frac{p_{ij}}{{N}_{ij}}$ and ${N}_{ij}$ are the normalisation constant arising from Dicke state $\ket{D_{ij}}$, i.e.,
 {\begin{align*}
    N_{ij}= \left \{ \begin{array}{rl}
    1 \quad \mbox{for} \quad i=j~,\\
    2 \quad \mbox{for} \quad i\neq j~.
    \end{array}
    \right.
\end{align*} }
\end{proposition}
\noindent Clearly, positivity of  $\rho_{DS}^{\Gamma}$ relies on the positivity of  $M_d(\rho_{DS})$. Notice also that due to symmetry, it is irrelevant w.r.t. which party the partial transposition is performed. Now, the following matrix algebra definitions are needed. 
\begin{definition}
    A $d \times d$ matrix, M, is said completely positive if and only if there exists a non-negative $d \times k$ matrix, $C$, such that $M=CC^T$, for some $k\geq1$. This set of matrices forms a convex cone, dubbed $\mathcal{CP}_d$. 
\end{definition}

\begin{definition}
    A $d \times d$ matrix, $M$, is said to be doubly non-negative if and only if it is positive semidefinite and its entries are non-negative. This set of matrices forms a convex cone, dubbed $\mathcal{DNN}_d$.
\end{definition}

\noindent With the aid of the above definitions, the following theorems follow swiftly:
\begin{theorem}[\onlinecite{yu2016separability,tura2018separability}]
\label{CP}
     Let $\rho_{DS} \in  \mathcal{B}(\mathcal{S(}\mathbb{C}^d\otimes \mathbb{C}^d))$ be a DS state. Then,
      $\rho_{DS}\text{ separable} \iff M_d(\rho_{DS}) \in \mathcal{CP}_d$.
\end{theorem}
\begin{theorem}[\onlinecite{tura2018separability}]
    Let $\rho_{DS} \in  \mathcal{B}(\mathcal{S(}\mathbb{C}^d\otimes \mathbb{C}^d))$ be a DS state. Then,
$      \rho_{DS}^{\Gamma}\geq 0 \;  \iff M_d(\rho_{DS}) \in \mathcal{DNN}_d.$
\end{theorem}
 
\noindent Interestingly, the convex cones $\mathcal{CP}_{d}$ and $\mathcal{DNN}_{d}$ are related by the following Lemma \onlinecite{diananda1962non}:
\begin{lemma}
\label{equality}
The cone of doubly non-negative matrices and the cone of completely positive matrices coincide for $d \leq 4$, i.e., $\mathcal{CP}_d = \mathcal{DNN}_d$.
\end{lemma}
\noindent Hence, separability for bipartite DS states can be cast as:
\begin{theorem}[\onlinecite{yu2016separability,tura2018separability}]
Let $\rho_{DS} \in  \mathcal{B}(\mathcal{S(}\mathbb{C}^d\otimes \mathbb{C}^d))$ be a DS state with $d\leq 4$. Then,
    \begin{equation}
        \rho_{DS}\text{ separable} \iff \rho_{DS} \; \text{PPT}.
    \end{equation}
    \label{bipartite}
\end{theorem}

\noindent Explicit examples of PPT-entangled DS states for $d\geq 5$ have been provided in \onlinecite{tura2018separability}, along with the entanglement witnesses that are able to detect them \onlinecite{marconi2021entangled}. 

With all these tools in hand, we present now our results for multipartite symmetric states together with some open conjectures.

\section{\label{sec:3} Entanglement in multipartite diagonal symmetric states}

\subsection{Entanglement in tripartite diagonal symmetric states}

\noindent We start by demonstrating that Theorem \ref{bipartite}, holds true also in $N=3$ for any local dimension $d$, i.e., for all $\rho_{DS} \in \mathcal{B}(\mathcal{S}((\mathbb{C}^d)^{\otimes 3}))$.

Analogously to the bipartite case, where the relevant information of the partial transpose is encoded in the matrix $M_d(\rho_{DS})$ (see Eqs. (\ref{eqmdef})-(\ref{mbipartit})), also for the tripartite case it is possible to introduce an equivalent matrix, $M^{(3)}(\rho_{DS})$, as stated by the following Lemma. 

\begin{lemma}
    Let $\rho_{DS} \in \mathcal{B}(\mathcal{S}((\mathbb{C}^d)^{\otimes 3}))$. Then, its partial transposition w.r.t. any party leads to a $d^2 \times d^2$ matrix, $M^{(3)}(\rho_{DS})$, of the form:
     \begin{align}
     \label{ds:n3}
 M^{(3)}(\rho_{DS})= \bigoplus_{i=0}^{d-1} M^{(3)}_i(\rho_{DS})~,\hspace{20pt} \\
 M^{(3)}_i (\rho_{DS})= \begin{pmatrix}
        \bar{p}_{i,0,0} & \bar{p}_{i,0,1} & \cdots&\bar{p}_{i,0, d-1}\\
        \bar{p}_{i,0,1} & \bar{p}_{i,1,1} & \cdots&\bar{p}_{i,1,d-1} \\
        \vdots&\vdots&\ddots&\vdots\\
         \bar{p}_{i,0, d-1} & \bar{p}_{i,1, d-1} & \cdots& \bar{p}_{i, d-1 ,d-1}
         \label{miforma}
    \end{pmatrix}
\end{align}
with $\bar{p}_{ijk}=\frac{p_{ijk}}{{N}_{ijk}}$ and ${N_{ijk}}$ being the normalisation factor of the corresponding tripartite Dicke state, i.e.,
\begin{align*}
    N_{ijk}= \left \{ \begin{array}{rl}
1 \quad \mbox{for} \quad i=j=k~,\\
3 \quad \mbox{for} \quad i=j \neq k~,\\
6 \quad \mbox{for} \quad i \neq j \neq k~.
\end{array}
\right.
\end{align*} 
\label{estat}
\end{lemma}
Notice that, in this case, each matrix $M_i^{(3)}(\rho_{DS})$ lives in the subspace spanned by the elements of the computational basis $\{\ket{i00},\ket{i11},\dots,\ket{idd}\}$ and, therefore, has at most rank $d$.
\noindent Here, the superindex $(3)$ reminds that we are considering a tripartite system. 
\noindent  Moreover, as $M^{(3)}(\rho_{DS})$ in Eq.(\ref{ds:n3}) is a direct sum of $d$ matrices of order $d$, the entanglement properties of $\rho_{DS}$ are determined by the latter set of matrices. For such reason, we can rely on the results presented in Section \ref{sec:2} to derive the following theorems: 
\begin{theorem}
\label{mr1}
Let $\rho_{DS} \in \mathcal{B}(\mathcal{S}((\mathbb{C}^d)^{\otimes 3}))$ be DS with associated matrix $M^{(3)}(\rho_{DS})$. Then,
\begin{center}
    $\rho_{DS}$ separable $\iff M^{(3)}(\rho_{DS}) \in \mathcal{CP}_{d^2}$~.
\end{center}
\end{theorem}

\begin{theorem}
\label{mr2}
Let $\rho_{DS} \in \mathcal{B}(\mathcal{S}((\mathbb{C}^d)^{\otimes 3}))$ be DS with associated matrix $M^{(3)}(\rho_{DS})$. Then,
\begin{center}
    $\rho_{DS}$ PPT $\Longleftrightarrow M^{(3)}(\rho_{DS}) \in \mathcal{DNN}_{d^2}$~.
\end{center}
\end{theorem}
For the sake of easiness, proofs of Theorems \ref{mr1}-\ref{mr2} are moved to appendices \ref{p1} and \ref{p2}, together with the explicit construction of $M^{(3)}(\rho_{DS})$. Moreover, the following observation is needed:
\begin{observation}
    Let $M$ be a square matrix with direct sum structure, i.e., $M=\bigoplus_{k} M_{k}$. Then, M is completely positive (double non-negative) if and only if, for every $k$, $M_{k}$ is completely positive (doubly non-negative).
    \label{propsum}
\end{observation}
The above theorem now directly leads to our first result
\begin{theorem}
  \label{chachi}
Let $\rho_{DS} \in \mathcal{B}(\mathcal{S}((\mathbb{C}^d)^{\otimes 3}))$ with $d \leq 4$. Then, 
\begin{center}
    $\rho_{DS}$ separable $\Longleftrightarrow \rho_{DS}$ PPT.
\end{center}

\end{theorem}
Theorem \ref{chachi}, ensures that there is no PPT-entanglement (bound entanglement) in DS states of three parties when the local dimension $d\leq 4$. However, for $d\geq 5$, there must exist tripartite diagonal symmetric states that are PPT-entangled: any state whose tracing out one party leads to a bipartite PPT-entangled state is necessary PPT entangled.


\subsection{Entanglement in N-partite DS states for \texorpdfstring{$N\geq 4$}{N}}

\noindent We first tackle the case $N=4$. It is now necessary to consider PPT conditions with respect to two different partitions, namely $2:2$ and $1:3$. First we prove that, for DS states, positivity w.r.t. the largest partition implies positivity w.r.t. any smaller partition. Even though the $N=4$ case, can be cast in a similar way to $N=3$, the dimension of the corresponding $M^{(4)}_{i}$ matrices prevents from ensuring that PPT implies separability for any $d$. Nevertheless, we extend Theorems \ref{mr1}-\ref{mr2} to $N=4$, and, we conjecture that they hold true for any $N$ and $d$. 

\begin{lemma}
\label{lemma:4partite}
    Let $\rho_{DS} \in \mathcal{B}(\mathcal{S}((\mathbb{C}^d)^{\otimes 4}))$. Then, its partial transposition w.r.t.  the largest partition, i.e., $\rho_{DS}^{\Gamma_{2:2}}$, leads to a $\frac{d^3+d}{2} \times \frac{d^3+d}{2} $ matrix $M^{(4)}(\rho_{DS})$ of the form:
    \begin{equation}
    \label{def4parties}
        M^{(4)}(\rho_{DS})= \begin{pmatrix}
            \bar{p}_{r_0,r_0}& \bar{p}_{r_0,r_1}&\dots&\bar{p}_{r_0,r_s}\\\bar{p}_{r_0,r_1}&\bar{p}_{r_1,r_1}&\dots&\bar{p}_{r_1,r_s}\\
            \vdots&\vdots&\ddots&\vdots\\
\bar{p}_{r_0,r_s}&\bar{p}_{r_1,r_s}&\dots&\bar{p}_{r_s,r_s}
        \end{pmatrix}
        \bigoplus_{\substack{i < j\\ i=0}}^{d-1}
        \begin{pmatrix}
            \bar{p}_{ij00}&\bar{p}_{ij01}& \dots &\bar{p}_{ij0 d-1}\\
\bar{p}_{ij01}&\bar{p}_{ij11}&\dots&\bar{p}_{ij1 d-1}\\
            \vdots&\dots&\ddots&\vdots\\
            \bar{p}_{ij0d-1}&\bar{p}_{ij1d-1}&\dots&\bar{p}_{ijd-1 d-1}
        \end{pmatrix}~,
    \end{equation}
   where $r_{i}$ denotes the $i$-th element of a list of dimension $s = d(d+1)/2$ corresponding to the ordered indices of the two-qudit Dicke states $\ket{D_{ij}}$.
\end{lemma}

\noindent For example, when $d=3$, $s=6$ and $r_{i} \in \{00,01,02,11,12,22\}$. We mention that $M^{(4)}(\rho_{DS})$ is built following the same procedure used in the bipartite and tripartite cases (see Appendix \ref{p2}).  From the structure displayed in Eq.(\ref{def4parties}) stems the following observation:
\begin{lemma}
    Let $\rho_{DS} \in \mathcal{B}(\mathcal{S}((\mathbb{C}^d)^{\otimes 4}))$. Then, $\rho_{DS}^{\Gamma_{2:2}} \succeq 0 \implies \rho_{DS}$ PPT.
    \label{lemman=4}
\end{lemma}
\begin{proof}
    The lemma can be proved by looking at the general structure of the partial transposition with respect to partition $1:3$, $\rho_{DS}^{\Gamma_{1:3}}$, which takes the form:
    \begin{equation}
    \rho_{DS}^{\Gamma_{1:3}}=
        \bigoplus_{\substack{i \leq j\\i,j=0}}^{d-1}\begin{pmatrix}
            \bar{p}_{ij00}&\bar{p}_{ij01}& \dots &\bar{p}_{ij0 d-1}\\
            \bar{p}_{ij01}&\bar{p}_{ij11}&\dots&\bar{p}_{ij1 d-1}\\
            \vdots&\dots&\ddots&\vdots\\
            \bar{p}_{ij0d-1}&\bar{p}_{ij1d-1}&\dots&\bar{p}_{ijd-1 d-1}
        \end{pmatrix}
        \label{1versus3}
    \end{equation}
Comparing Eq.(\ref{def4parties}) to Eq.(\ref{1versus3}), it can be observed that, choosing $i<j$, the matrices that arise in Eq.(\ref{1versus3})  are already included in $M^{(4)}(\rho_{DS})$.
Taking $i=j$, the related matrices correspond to principal minors of the first matrix of  Eq.(\ref{def4parties}), thus concluding the proof.
\end{proof}
The above Lemma ensures that we can focus only on $\rho_{DS}^{\Gamma_{2:2}}$ when inspecting the PPT of a given state, since the PPT conditions are already included in $M^{(4)}(\rho_{DS})$. When considering $N$ parties, we can generalise the previous matrices and deduce, by induction, the structure of the matrix $M^{(N)}(\rho_{DS})$, which arises from the partial transpose with respect to the largest partition:

\begin{lemma}
    Let $\rho_{DS} \in \mathcal{B}(\mathcal{S}((\mathbb{C}^d)^{\otimes N}))$ be a DS state with $N$ even. Then, its partial transpose w.r.t. the largest partition, i.e., $\rho_{DS}^{\Gamma_{N/2:N/2}}$, leads to the following associated matrix $M^{(N)}(\rho_{DS})$:
    \begin{align}
    \label{mgeneraleven}
        &M^{(N)}(\rho_{DS})=\bigoplus_{i=0}^{N/2-1}\left(\bigoplus_{k=1}^{k=t}M_k^i\right)~,\\ 
         &M^i_k=\begin{pmatrix}
            \bar{p}_{r_k,s_0,s_0}& \bar{p}_{r_k,s_0,s_1}&\dots&\bar{p}_{r_k,s_0,s_l}\\\bar{p}_{r_k,s_0,s_1}&\bar{p}_{r_k,s_1,s_1}&\dots&\bar{p}_{r_k,s_1,s_l}\\
            \vdots&\vdots&\ddots&\vdots\\
\bar{p}_{r_k,s_0,s_l}&\bar{p}_{r_k,s_1,s_l}&\dots&\bar{p}_{r_k,s_l,s_l}
        \end{pmatrix}    
    \end{align} 
 where $M^i_k$ are square matrices of dimension $\binom{\frac{N}{2}-i+d-1}{d-1}$, $r_k, s_k$ denote the $k$-th element of the ordered lists containing the indices of the Dicke states that span $\mathcal{S}((\mathbb{C}^d)^{\otimes 2i})$ with $i$ equal indices ($t$ different elements) and $\mathcal{S}((\mathbb{C}^d)^{\otimes \frac{N}{2}-i}))$, respectively. Here, $s_l$ denotes the last element of the list.

\end{lemma}

\begin{lemma}
    Let $\rho_{DS} \in \mathcal{B}(\mathcal{S}((\mathbb{C}^d)^{\otimes N}))$ be a DS state with $N$ odd. Then, its partial transpose w.r.t. the largest partition, i.e., $\rho_{DS}^{\Gamma_{\frac{N-1}{2}:\frac{N+1}{2}}}$, leads to the following associated matrix $M^{(N)}(\rho_{DS})$:
    \begin{align}
        &M^{(N)}(\rho_{DS})=\bigoplus_{i=0}^{\frac{N-1}{2}-1}\left(\bigoplus_{k=1}^{k=t}\left(\bigoplus_{j=0}^{j=d-1}M_{jk}^i\right)\right)~,\\ 
        &M^i_{jk}=\begin{pmatrix}
            \bar{p}_{j,r_k,s_0,s_0}& \bar{p}_{j,r_k,s_0,s_1}&\dots&\bar{p}_{j,r_k,s_0,s_l}\\\bar{p}_{j,r_k,s_0,s_1}&\bar{p}_{j,r_k,s_1,s_1}&\dots&\bar{p}_{j,r_k,s_1,s_l}\\
            \vdots&\vdots&\ddots&\vdots\\
\bar{p}_{j,r_k,s_0,s_l}&\bar{p}_{j,r_k,s_1,s_l}&\dots&\bar{p}_{j,r_k,s_l,s_l}
        \end{pmatrix}    
\label{mgeneralodd}
    \end{align} 
 where $M^i_{jk}$ are square matrices of dimension $\binom{\frac{N-1}{2}-i+d-1}{d-1}$, $r_k, s_k$ denote the $k$-th element of the ordered lists containing the indices of the Dicke states that span $\mathcal{S}((\mathbb{C}^d)^{\otimes 2i})$ with $i$ equal indices ($t$ different elements) and $\mathcal{S}((\mathbb{C}^d)^{\otimes \frac{N-1}{2}-i})$, respectively. Here, $s_l$ denotes the last element of the list.

\end{lemma}

\noindent Given the explicit structure of Eqs.(\ref{mgeneraleven})-(\ref{mgeneralodd}), consisting of nested direct sums of matrices $M^i$, we pose the following conjecture:
\begin{conjecture}
     Let  $\rho_{DS} \in \mathcal{B}(\mathcal{S}((\mathbb{C}^d)^{\otimes N}))$ be an $N$-partite DS state. Then, PPT w.r.t. to its largest partition ensures PPT w.r.t any other partition, i.e.,
    \begin{align*}
    &\rho_{DS}^{\Gamma_{N/2:N/2}} \succeq 0 \implies \rho_{DS} \mbox{ PPT for even } N \\
    &\rho_{DS}^{\Gamma_{\frac{N-1}{2}:\frac{N+1}{2}}} \succeq 0 \implies \rho_{DS} \mbox{ PPT for odd } N.
    \end{align*}   
    \label{conjecture}
\end{conjecture}

 \noindent While it has not been possible to prove that, in general, PPT w.r.t.~the largest partition implies PPT also w.r.t. any other partition, we provide strong analytical evidence that suggests that the above conjecture holds true. Neverless, regardless of the above conjecture, Lemma \ref{lemman=4} allows to extend the validity of Theorems \ref{mr1} and \ref{mr2} to the case of four-partite systems.
 
\begin{theorem}
\label{mr1n}
Let $\rho_{DS} \in \mathcal{B}(\mathcal{S}((\mathbb{C}^d)^{\otimes 4}))$ be a DS state with an associated matrix $M^{(4)}(\rho_{DS})$ which arises from the partial transpose w.r.t. the largest partition. Then, 
\begin{center}
    $\rho_{DS}$ separable $\Longleftrightarrow M^{(4)}(\rho_{DS}) \in \mathcal{CP}_{\frac{d^3+d}{2}}$~.
\end{center}
\end{theorem}

\begin{theorem}
\label{mr2n}
Let $\rho_{DS}\in\mathcal{B}(\mathcal{S((}\mathbb{C}^d)^{\otimes 4}))$ be a DS state with an associated matrix $M^{(4)}(\rho_{DS})$ which arises from the partial transpose w.r.t. the largest partition. Then,
\begin{center}
    $\rho_{DS}$ PPT $\Longleftrightarrow M^{(4)}(\rho_{DS}) \in \mathcal{DNN}_{\frac{d^3+d}{2}}$~.
\end{center}
\end{theorem}
Prooving Theorems \ref{mr1n}-\ref{mr2n} follows the proofs for Theorems \ref{mr1}-\ref{mr2} stated in the Appendix. The only significant difference is represented by the implication in the left direction of Theorem \ref{mr1n}, for which we refer the reader to Appendix \ref{mr1nproof}.

\noindent Further, we conjecture that Theorems \ref{mr1n}-\ref{mr2n} hold true for any $N$. This conjecture is further supported in the case of even $N$ by the embedding presented in the following section. Nevertheless, since already for $N=4$ the associated matrices have size $m\times m$, with $m>4$, Lemma \ref{equality} does not ensure that PPT is necessary and sufficient for separability, not even in the low-dimensional case, i.e., $d=3,4$.

\section{\label{sec:4}Mapping multipartite DS-states to bipartite symmetric states}

\noindent Increasing the number of parties to $N>4$ leads to matrices that cannot be analitycally handled. However, we show that any multipartite DS state of $N$ (even) qudits of local dimension $d$, can be mapped to a a bipartite symmetric state of local dimension $k>d$: $\rho_{DS}^{(N,d)} \Leftrightarrow \rho_S^{(2,k)} (Fig \ref{Map}) $. 

This mapping is based on the following steps: 
\begin{itemize}
    \item[1.] Any Dicke states of $N$ qudits can be expressed as a linear combination of tensor products of two Dicke states of $N/2$ qudits;
    \item[2.] Since $\mathcal{S}((\mathbb{C}^{d})^{\otimes N/2}) \cong \mathbb{C}^{k}$, with $k=\binom{N/2+d-1}{d-1}$, each Dicke state of $N/2$ qudits is associated to a basis vector in $\mathbb{C}^{k}$;
    \item[3.] Any Dicke state of $N$ qudits can now be expressed as a linear combination of \textit{bipartite} Dicke states of local dimension $k$.
    \item[4.] Hence, the initial state $\rho_{DS}$ is transformed into a symmetric state $\rho_S \in \mathcal{B}(\mathcal{S}(\mathbb{C}^k\otimes \mathbb{C}^k))$\label{punt4}.
\end{itemize}
It is worth stressing that this correspondence does not define an isomorphism between the two subspaces, since the first has dimension $\binom{N+d-1}{d-1}$, while the latter has dimension $k(k+1)/2=\binom{N/2+d-1}{d-1}\left(\binom{N/2+d-1}{d-1}+1\right)/2$, i.e., significantly greater. As a consequence, the resulting state through the mapping is generically a symmetric, although not DS, state.
One may object that it might be more convenient to map a multipartite DS state into another symmetric state by means of an isomorphism, rather than considering an embedding into a higher dimensional subspace. However, it is necessary to take into account both the state and its partial transpositions with respect to different partitions. Let us clarify this point with an example. Since $\mbox{dim } \mathcal{S}\left( (\mathbb{C}^{3})^{\otimes 4} \right) = 15 = \mbox{dim } \mathcal{S}\left( \mathbb{C}^{5} \otimes \mathbb{C}^{5} \right)$, there exists an isomorphism between the two subspaces. Hence, one might be tempted to associate a state $\rho_{DS}\in \mathcal{S}\left( (\mathbb{C}^{3})^{\otimes 4} \right) $ to a bipartite DS state $\tilde{\rho}_{DS} \in \mathcal{S}\left( \mathbb{C}^{5} \otimes \mathbb{C}^{5} \right)$. However, since $\mbox{rank} (\rho_{DS}^{\Gamma_{2:2}}) > \mbox{rank} (\tilde{\rho}_{DS}^{\Gamma})$, such a naive isomorphism is not able to encode the whole information of the state $\rho_{DS}$.

\begin{figure}
    \centering
    \includegraphics[scale=0.4]{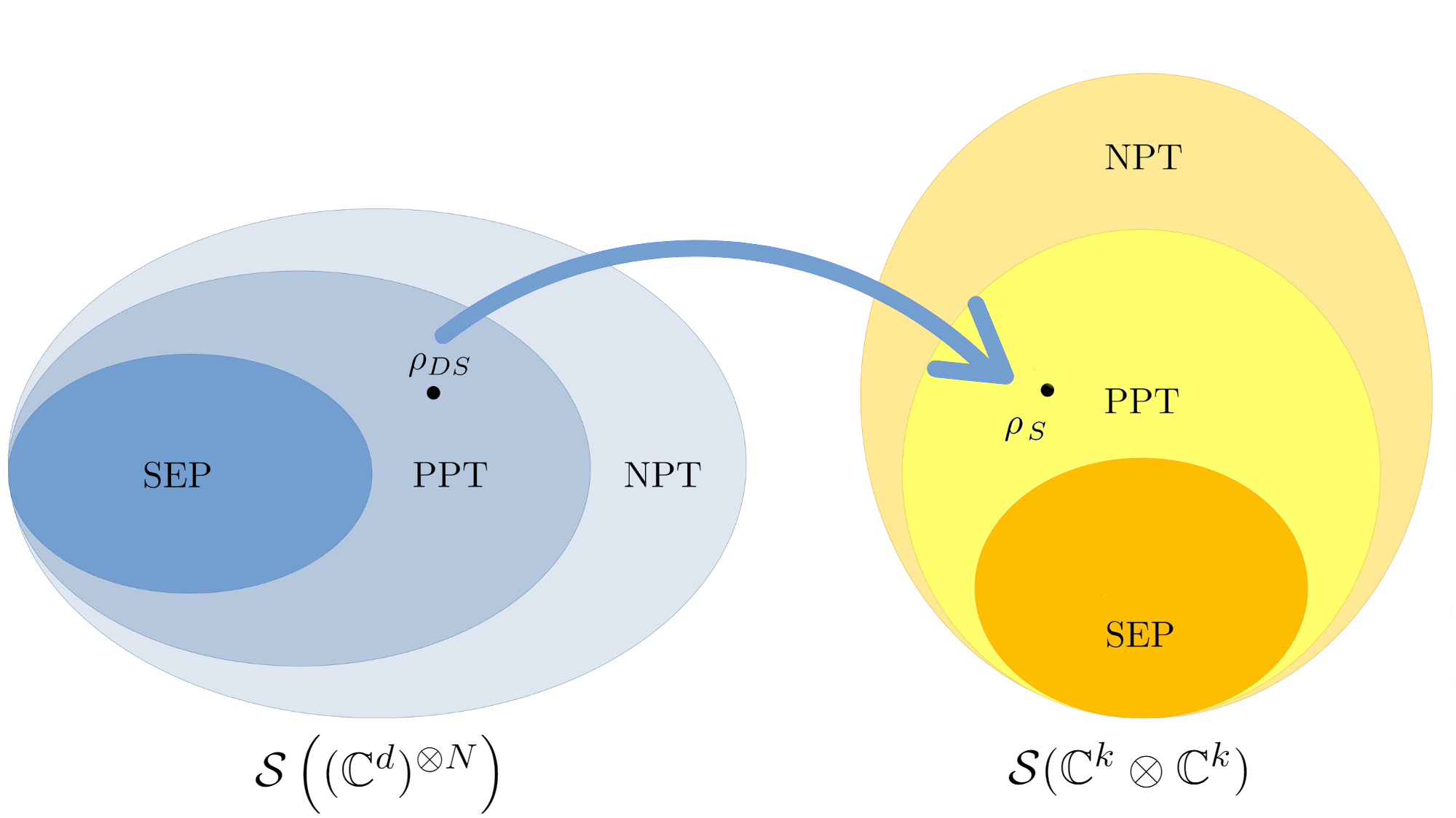}
    \caption{Pictorial representation of the embedding of a multipartite DS state $\rho_{DS}$ into a bipartite symmetric state $\rho_{S}$ of systems with higher local dimension, i.e., $k>d$.}
    \label{Map}
\end{figure}
Notice that this embedding is optimal in the sense that $k$ corresponds to the smallest local dimension of the bipartite system such that $\mbox{rank}(\rho_{DS}) < \mbox{rank}(\rho_{S})$ and $\mbox{rank}(\rho_{DS}^{\Gamma_{N/2:N/2}}) = \mbox{rank}(\rho_{S}^{\Gamma})$. This implies that, in order to check the conditions of $\rho_{DS}$, it is now sufficient to inspect the bipartite state $\rho_S$ which results from the embedding. This is due to the fact that the PPT conditions on this bipartite symmetric state that results from the embedding, are identical to the PPT conditions arising in the multipartite DS case. This fact strongly supports that for even $N$, both the conjecture \ref{conjecture} and the Theorems \ref{mr1n}-\ref{mr2n} hold true.  The advantage of this approach becomes particularly evident when one compares the dimension of the partial transpositions, since $\mbox{dim}(\rho_{S}^{\Gamma}) = \binom{N/2+d-1}{d-1}^2 < \mbox{dim}(\rho_{DS}^{\Gamma_{N/2:N/2}})=d^{N}$, as it can be seen in Fig.\ref{ptdim}.
We expect the same behaviour for odd $N$, although this direct embedding cannot be applied in this case. 
\begin{figure}[t]
    \centering
    \includegraphics[scale=0.8]{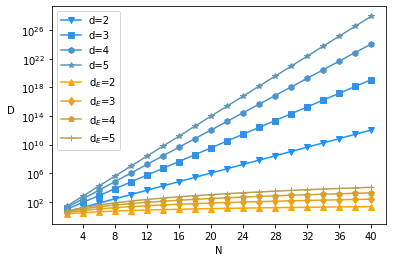}
    \caption{Scaling of the dimension $(D)$ of
   $(\rho_{DS})^{\Gamma_{N/2:N/2}}$ and $\rho_{S}^{\Gamma}$  
   for an even number of parties $(N)$. Here, $\rho_{DS}$ is a DS state with $N$ parties of dimension $d$ (blue) and $\rho_{S}$ is a bipartite symmetric state that results from the mapping of a DS state of local dimension $d_E$ (orange).}
    \label{ptdim}
\end{figure}

\noindent Let us illustrate the above embedding with an explicit example, where we map a multipartite DS state with $N=4$ and $d=3$ to a bipartite symmetric state with $d=6$, i.e., $\rho_{DS}\in \mathcal{B}\left(\mathcal{S}\left( (\mathbb{C}^3)^{\otimes 4} \right)\right) \rightarrow \rho_{S}\in \mathcal{B}\left(\mathcal{S}\left( \mathbb{C}^6 \otimes \mathbb{C}^6\right)\right)$. First, notice that the Dicke basis for four qutrits can be expressed as a linear combination of tensor products between Dicke states of two qutrits, i.e.,
\begin{align}
    &\ket{D^{(3)}_{iiii}}=\ket{D^{(3)}_{ii}}\ket{D^{(3)}_{ii}}\\
    &\ket{D^{(3)}_{iiij}}=\frac{1}{\sqrt{2}}\left(\ket{D^{(3)}_{ii}}\ket{D^{(3)}_{ij}}+\ket{D^{(3)}_{ij}}\ket{D^{(3)}_{ii}}\right)\\
    &\ket{D^{(3)}_{iijj}}=\frac{1}{\sqrt{3}}\left(\ket{D^{(3)}_{ii}}\ket{D^{(3)}_{jj}}+\ket{D^{(3)}_{jj}}\ket{D^{(3)}_{ii}}\right)+\sqrt{\frac{2}{3}}\ket{D^{(3)}_{ij}}\ket{D^{(3)}_{ij}}\\
    &\ket{D^{(3)}_{iijk}}=\frac{1}{\sqrt{6}}\left(\ket{D^{(3)}_{ii}}\ket{D^{(3)}_{jk}}+\ket{D^{(3)}_{jk}}\ket{D^{(3)}_{ii}}\right)+\\
    &\quad \sqrt{\frac{1}{3}}\left(\ket{D^{(3)}_{ij}}\ket{D^{(3)}_{ik}}+\ket{D^{(3)}_{ik}}\ket{D^{(3)}_{ij}}\right)
\end{align}
with $i,j,k \in \{0,1,2\},~i\neq j \neq k$ and where the superscript represents qutrits. One must note that, in general we would still be missing one possible Dicke state, i.e.; the Dicke state with four different indices. Nevertheless, as this specific example consists of qudits, we can, at most have 3 different indices.
Recalling that $\mathcal{S}(\mathbb{C}^{3} \otimes \mathbb{C}^{3}) \cong \mathbb{C}^{6}$, we can relabel them as:
\begin{align}
    &\ket{D^{(3)}_{00}}=\ket{0}, \quad \ket{D^{(3)}_{01}}=\ket{1}, \quad \ket{D^{(3)}_{02}}=\ket{2}, \\&\ket{D^{(3)}_{11}}=\ket{3}, \quad \ket{D^{(3)}_{12}}=\ket{4}, \quad \ket{D^{(3)}_{22}}=\ket{5}.
\end{align}
where $\{\ket{0},\dots, \ket{5}\}$ refers to the computational basis of $\mathbb{C}^6$. As a consequence, a generic DS state in $\mathcal{S}((\mathbb{C}^3)^{\otimes 4})$, leads to a bipartite symmetric state $\rho_S^{(6)} \in \mathcal{B}\left(\mathcal{S}\left( \mathbb{C}^6 \otimes \mathbb{C}^6\right)\right)$ of the form
\begin{align}
\label{sym6}
    \rho_S^{(6)}=\sum_{ij} p_{ij} \ketbra{D_{ij}^{(6)}}{D_{ij}^{(6)}}+ \sum_{ijkl} \alpha_{ij}^{kl} \ketbra{D_{ij}^{(6)}}{D_{kl}^{(6)}}+ \mbox{h.c.}~
\end{align}

\noindent Notice that, although the size of the initial state generically increases, the dimension of the partial transposition of the mapped state becomes considerably smaller as compared to the one of the initial state, thus greatly simplifying the characterisation of PPT-entanglement. For instance, the initial four-partite qutrit state $\rho_{DS}$ from the previous example has ranks $ (\mbox{rank} \rho_{DS}, \mbox{rank} \rho_{DS}^{\Gamma_{2:2}}) = (15,81)$, while the bipartite qudit state $\rho_{S}$ with $d=6$ has ranks $ (\mbox{rank} \rho_{S}, \mbox{rank} \rho_{S}^{\Gamma}) = (21, 36)$. 

\section{\label{sec:5} Mapping symmetric states of qudits to DS states of qubits}

\noindent In this Section we present a technique to recast multipartite qudit systems as DS states of qubits. Noticing that $\mathcal{S}((\mathbb{C}^{2})^{\otimes N}) \cong \mathbb{C}^{N+1}$ , one can set $\ket{D^{N}_{i}} \equiv \ket{i}$, where $\ket{D^{N}_{i}}$ is the Dicke state of $N$ qubits with $i$ excitations, and $\ket{i} \in \mathbb{C}^{N+1}$. Once the mapping has been performed, exploting the fact that PPT DS states of qubits are separable \onlinecite{Quesada_2017}, we can deduce equivalent properties for the original qudit state. For the sake of clarity, in the following we restrict to bipartite states, although we stress that our approach can be generalised to any number of parties. Let us consider the symmetric space of two qutrits, i.e., $\mathcal{S}\left(\mathbb{C}^{3} \otimes \mathbb{C}^3 \right)$. Notice that $\mathbb{C}^{3} \cong \mathcal{S}\left( (\mathbb{C}^{2})^{\otimes 2} \right)$, since it is possible to write
\begin{equation}
\label{map:qutrit}
    \ket{0} = \ket{D^{2}_{0}}, \quad 
    \ket{1}= \ket{D^{2}_{1}}, \quad 
    \ket{2}= \ket{D^{2}_{2}}~.
\end{equation}

\noindent Hence, 
\begin{equation}
 \mathcal{S}\left(\mathbb{C}^{3} \otimes \mathbb{C}^3 \right) = \mathcal{S}\left( \mathcal{S}((\mathbb{C}^{2})^{\otimes 2}) \otimes \mathcal{S}((\mathbb{C}^{2})^{\otimes 2}) \right) \supset \mathcal{S}((\mathbb{C}^{2})^{\otimes 4}))~, 
\end{equation}

\noindent where the last inclusion implies that every Dicke state $\ket{D^{4}_{i} }\in \mathcal{S}((\mathbb{C}^{2})^{\otimes 4})$ can be expressed as a linear combination of the Dicke states $\ket{D^{(3)}_{ij}} \in \mathcal{S}(\mathbb{C}^{3} \otimes \mathbb{C}^{3})$, i.e.,
\begin{align}
\label{map:qub1}
&\ket{D_0^4}=\ket{D^{(3)}_{00}}, \quad
\ket{D_1^4}=\ket{D^{(3)}_{01}}, \\
\label{map:qub2}
&\ket{D_2^4}=\frac{1}{\sqrt{3}}\ket{D^{(3)}_{02}}+\sqrt{\frac{2}{3}}\ket{D^{(3)}_{11}}, \\
\label{map:qub3}
&\ket{D_3^4}=\ket{D^{(3)}_{12}}, \quad
\ket{D_4^4}=\ket{D^{(3)}_{22}}~.
\end{align}

\noindent Inverting Eqs.(\ref{map:qub1})-(\ref{map:qub3}) leads to
\begin{align}
\label{map:qud1}
&\ket{D^{(3)}_{00}} = \ket{D_0^4}, \quad
\ket{D^{(3)}_{01}} = \ket{D_1^4}, \\
\label{map:qud2}
&\ket{D^{(3)}_{02}} = \sqrt{3}\ket{D_2^4} - \sqrt{2} \ket{D^{(3)}_{11}}, \\
\label{map:qud3}
&\ket{D^{(3)}_{12}}= \ket{D_3^4}, \quad
\ket{D^{(3)}_{22}}= \ket{D_4^4}~,
\end{align}
\noindent from which it is evident that not every state in $ \mathcal{S}(\mathbb{C}^{3} \otimes \mathbb{C}^{3})$ can be mapped to a state in $\mathcal{S}((\mathbb{C}^{2})^{\otimes 4}))$. However, imposing conditions on a state $\rho_{S} \in \mathcal{S}(\mathbb{C}^{3} \otimes \mathbb{C}^{3})$ such that $\mbox{rank }(\rho_{S}) = 5$, allows to define an isomorphism to a four-qubit DS state $\rho^{Q}_{DS} = \sum_{k} q_{k} \ketbra{D^{4}_{k}}{D^{4}_{k}}$. In fact, consider the symmetric state $\rho_{S}$ given by
\begin{equation}
\label{one_guhne}
     \rho_{S} = \rho_{DS} + (\alpha^{11}_{02} \ketbra{D_{11}^{(3)}}{D_{02}^{(3)}}+ \mbox{h.c.})~.
\end{equation}
\noindent Imposing $p_{11}=2p_{02} = \sqrt{2} \alpha^{11}_{02} $ and using Eqs.(\ref{map:qud1})-(\ref{map:qud3}), $\rho_{S}$ of Eq.(\ref{one_guhne}) is mapped to a DS state $\rho^{Q}_{DS}$ whose coefficients $q_{k}$ satisfy
\begin{align}
\label{cond:sep1}
    &q_{0} =  p_{00},  \;\;\quad   q_{1}= p_{01}, \\
    \label{cond:sep2}
    &q_{2} = 3p_{02},  \quad q_{3} =p_{12}, \quad  q_{4}=p_{22}.
\end{align}

\noindent As a consequence of this mapping we have the following result: 
\begin{corollary}
\label{th:sep:map}
    Let $\rho_{S}\in \mathcal{S}(\mathbb{C}^{3} \otimes \mathbb{C}^{3})$ be a symmetric state of the form
    \begin{equation}
 \rho_{S} = \rho_{DS} + (\alpha^{11}_{02} \ketbra{D_{11}^{(3)}}{D_{02}^{(3)}} + \emph{h.c.})~, 
    \end{equation}
 whose coefficients satisfy $p_{11}=2p_{02} = \sqrt{2} \alpha^{11}_{02}$. Then, $\rho_{S}$ is separable iff it is PPT.
\end{corollary}
The proof of this corollary follows trivially from the fact that this type of states can be mapped to multipartite diagonal symmetric qubits, which are proven to be separable iff PPT \onlinecite{Quesada_2017,yu2016separability}. 
\noindent We conclude this section noting that the above result can be easily generalised to the case of other coherences of the same form, i.e., $\sigma_{CS} = \alpha^{ii}_{jk} \ketbra{D_{ii}^{(3)}}{D_{jk}^{(3)}} + \mbox{h.c.},$ for $i\neq j \neq k$, by considering an alternative mapping.\\

\section{\label{sec:conclusions}Conclusions}
\noindent The study of symmetric states hold several advantages as their  properties and robustness against noise, make them potential candidates to  simulation, computation and sensing.  Here, we have focus on the entanglement properties of  multipartite diagonal symmetric states (mixtures of projectors in generalized Dicke states) where, generically, much less about entanglement properties are known. First, we have demonstrated that, for $N=3$, PPT is necessary and sufficient for separability for $d=3,4$, and
presented strong evidence that this will be the case for any number of parties with these local dimensions. This conjecture is further supported by embedding multipartite DS states into bipartite symmetric states with higher local dimension. The advantage of this approach is twofold: on the one hand, the entanglement properties of multipartite systems can be deduced from known results regarding separability of bipartite symmetric states; on the other, the dimension of the partial transpose after the embedding is greatly reduced, providing a valuable tool in the characterisation of the entanglement in the multipartite scenario. Finally, we show that for local dimensions $d\geq 5$, there exist always bound entanglement irrespectively on the number of 
parties $N$.

\section{\label{sec:Acknowledgments}Acknowledgments}
\noindent J.R-P acknowledges financial support from Ministerio de Ciencia e Innovación of the Spanish Goverment FPU22/01511. J.A. and A.S. acknowledges financial support from Ministerio de Ciencia e Innovación of the Spanish Goverment with funding from European Union NextGenerationEU (PRTR-C17.I1) and by Generalitat de Catalunya. C.M. acknowledges support from the European Union - NextGeneration EU, "Integrated infrastructure initiative in Photonic and Quantum Sciences" - I-PHOQS [IR0000016, ID D2B8D520, CUP B53C22001750006] . A.S. acknowledges financial support from the European Commission QuantERA grant ExTRaQT (Spanish MICIN project PCI2022-132965), by the Spanish MICIN (project PID2022-141283NB-I00) with the support of FEDER funds, and by the Ministry for Digital Transformation and of Civil Service of the Spanish Government through the QUANTUM ENIA project call - Quantum Spain project, and by the European Union through the Recovery, Transformation and Resilience Plan - NextGeneration EU within the framework of the Digital Spain 2026 Agenda.

\bibliographystyle{ieeetr}
\bibliography{biblo.bib}

\section*{\label{sec:Author Declarations}Author Declarations}

\subsection*{Conflict of Interest}

\noindent The authors have no conflicts to disclose.

\subsection*{Author Contributions}

\noindent \textbf{Jordi Romero-Pallejà}: Investigation (lead); formal analysis (lead); conceptualization (equal); writing-original draft (lead); writing-review and editing (equal). \textbf{Jennifer Ahiable}:formal analysis (supporting). \textbf{Alessandro Romancino}: formal analysis (supporting). \textbf{Carlo Marconi}: Investigation (supporting); formal analysis (supporting); conceptualization (equal); writing-original draft (supporting); writing-review and editing (lead); supervision (equal). \textbf{Anna Sanpera}:Investigation (supporting); formal analysis (supporting); conceptualization (lead); writing-original draft (supporting); writing-review and editing (equal); supervision (equal).

\subsection*{Data Availability}

\noindent Data sharing is not applicable to this article as no new data were created or analyzed in this study.

\appendix
\section{Proof of Theorem \ref{mr1}}
\label{p1}

\begin{proof}

\noindent$(\Rightarrow$) Let us recall that any separable DS state $\rho_{DS}$ admits a decomposition of the form:
\begin{align}
\label{dsseparable}
    \rho_{DS}=\sum_i \lambda_i \ketbra{e_i}{e_i}^{\otimes 3}~,
\end{align}
with $\ket{e_i}=\sum_{j=0}^{d-1} e_{i,j}\ket{j}$ and $e_{i,j} \in \mathbb{C}$. Casting Eq.(\ref{dsseparable}) in the computational basis yields:
\begin{equation}
    \rho_{DS} = \sum_i \lambda_i \sum_{jklrst}^{d-1} e_{i,j} e_{i,k} e_{i,l} e_{i,r}^* e_{i,s}^* e_{i,t}^* \ketbra{jkl}{rst}~,
\end{equation}
from which, projecting onto the Dicke basis, leads to the following probabilities:
\begin{align}
    p_{aaa}&=\sum_i \lambda_i |e_{i,a}|^6~,\\
    p_{aab}&=\sum_i 3\lambda_i |e_{i,a}|^4|e_{i,b}|^2~,\\ 
    p_{abc}&=\sum_i 6\lambda_i |e_{i,a}|^2|e_{i,b}|^2 |e_{i,c}|^2~.
\end{align}

\noindent The matrix $M^{(3)}$ can be constructed realising that each submatrix $M^{(3)}_{i}$ (see Eq.(\ref{ds:n3}) in the main text) can be set as $M^{(3)}_{i} = \sum_{j} \lambda_j M^{(3)}_{ij}$, where
\begin{align}
M^{(3)}_{ij}&=|e_{j,i}|^2  \begin{pmatrix}
        |e_{j,0}|^4 & |e_{j,0}|^2|e_{j,1}|^2 & \dots& |e_{j,0}|^2 |e_{j,d-1}|^2\\
        |e_{j,0}|^2|e_{j,1}|^2 & |e_{j,1}|^4 & \dots& |e_{j,1}|^2 |e_{j,d-1}|^2\\
       \vdots &\vdots &\ddots&\vdots\\
        |e_{j,0}|^2 |e_{j,d-1}|^2&|e_{j,1}|^2 |e_{j,d-1}|^2 &\ldots& |e_{j,d-1}|^4
    \end{pmatrix}~.
\end{align}
\noindent Notice that, for the ease of readability, we have set $M^{(3)}_{i} \equiv M^{(3)}_{i}(\rho_{DS}) $ and equivalently for $M^{(3)}_{ij}$.
\noindent It can easily be seen that each matrix $M^{(3)}_{ij}$ is completely positive since it can be decomposed as $M^{(3)}_{ij}=B_i^j {B_i^j}^T$, with
\begin{equation}
    B_i^j=|e_{j,i}|\begin{pmatrix}
 |e_{j,0}|^2,|e_{j,1}|^2,\ldots, |e_{j,d-1}|^2
\end{pmatrix}^T~.
\end{equation}
Hence, $M^{(3)}_{i} \in \mathcal{CP}_{d}$ for every $i$, being a convex combination of completely positive matrices. As a consequence, recalling that $M^{(3)} = \bigoplus_{i=0}^{d-1} = M^{(3)}_i$, it follows that $M^{(3)} \in \mathcal{CP}_{d^2}$, being a direct sum of completely positive matrices. 

($\Leftarrow$) Let us now prove that $M^{(3)} \in \mathcal{CP}_{d^2}$ implies that $\rho_{DS}$ is separable. Since $M^{(3)}$ can be decomposed as $ M^{(3)}= \bigoplus_{i=0}^{d-1} M^{(3)}_i$, it follows that each matrix $M^{(3)}_i$ has to be completely positive, i.e.,
\begin{equation}
    M^{(3)}_i = B_i B_i^T=\sum_j \mathbf{b}_j^i{\mathbf{b}_j^i}^T
\end{equation}
where $\mathbf{b}_j^i$ are the columns of the matrix $B_i$.
Then, we can define quantum states from these vectors, i.e.,
\begin{equation}
    \ket{\chi_j^i}=\sum_{0\leq l <d}\chi_{j,l}^{i}\ket{l}, \quad |\chi_{j,l}^{i}|^2 =(\mathbf{b}_j^i)_l
\end{equation}

\noindent Notice, however, that inserting these states in Eq.(\ref{dsseparable}) does not reproduce a DS state. Hence, in order to erase the unwanted coherences between Dicke states, we introduce the following states instead: 
\begin{equation}
    \ket{\xi_{j,\mathbf{k},s}^i}=\sum_{0\leq l<d} (-1)^{\mathbf{k}_l}w^{sl}\chi^i_{j,l}\ket{l}~,
\end{equation}
with $\mathbf{k}$ being a vector of size $2^d$ in base 2 (i.e., the sum goes over all the possible vectors of size $2^d$ in base 2) and $w$ a $2d^{th}$-root of the unity, i.e., $w=e^{2\pi i/2d}$. Now, we can construct the density matrix associated to each of the terms that arise from the matrix $M^{(3)}_i=\sum_j M^{(3)}_{ij}$, i.e., the matrix $\rho_j^i$ that corresponds to the column $\mathbf{b}_j^i$. Similarly to the approach followed in \onlinecite{Quesada_2017} we construct bipartite separable states of the form:
\begin{equation}
    \rho_j^i=\sum_{0\leq s<d}\sum_{0\leq \mathbf{k}<2^d}\ket{\xi^i_{j,\mathbf{k},s}}\ket{\xi^i_{j,\mathbf{k},s}}\bra{\xi^i_{j,\mathbf{k},s}}\bra{\xi^i_{j,\mathbf{k},s}}
    \label{chaosdee}
\end{equation}
whose expression in the computational basis yields:
\begin{equation}
\begin{split}
        \rho_j^i=\sum_{\substack{0\leq \mathbf{k}<2^d}} \sum_{\substack{0\leq s<d\\
        0\leq l_1,l_2,l_3,l_4 <d}} (-1)^{k_{l_1}+k_{l_2}+k_{l_3}+k_{l_4}}\hspace{2pt} w^{s(l_1+l_2-l_3-l_4)}\chi^{i}_{j,l_1} \chi^{i}_{j,l_2} {\chi^{*}}^{i}_{j,l_3} {\chi^{*}}^{i}_{j,l_4} \ket{l_1}\ket{l_2}\bra{l_3}\bra{l_4}
    \label{muerte}
    \end{split}
\end{equation}
\noindent Performing the sum over $s$ in Eq.(\ref{muerte}) we find
\begin{equation}
    \sum_{0\leq s<d}w^{s(l_1+l_2-l_3-l_4)}= d \delta_{\substack{l_1+l_2,\\l_3+l_4}}
\label{iarrel}
\end{equation}
where the equality follows from the observation that the above sum is a geometrical series and $w$ is a $2d^{\mbox{th}}$-root of the unity. An immediate consequence of Eq.(\ref{iarrel}) is that the total number of excitations of the kets must obviously coincide with the total number of excitations of the bras. Now, we shall focus on the sum over $\mathbf{k}$, i.e.,
\begin{equation*}
    \sum_{0\leq \mathbf{k}<2^d} \substack{(-1)^{k_{l_1}+k_{l_2}+k_{l_3}+k_{l_4}}} 
\end{equation*}
Before proceeding, let us recall that $\mathbf{k}$ is a vector of size $2^d$ expressed in base 2. It can be seen that the only terms that survive in the above sum are those for which every $k_i$ is an even number. In order to make this statement clearer, it is useful to think of these indices as arranged in partitions of the form $(k_1,k_2,k_3,k_4)$. Hence, the only contributions different from zero are those corresponding to partitions (4,0,0,0) and (2,2,0,0) (2 indices have value $a$ and the other 2 have value $b\neq a$). In these cases, the sum yields a factor $2^d$. Combining this result with the condition found from Eq.(\ref{iarrel}), we can rewrite Eq.(\ref{muerte}) as:
\begin{equation}
\begin{split}
    \rho_j^i=d2^d\Bigl(\sum_{0\leq a< d}|\chi^i_{j,a}|^4\ketbra{D_{aa}}{D_{aa}}+\sum_{\substack{0\leq a,b <d,\\a\neq b}} 3 |\chi^i_{j,a}|^2|\chi^i_{j,b}|^2\ketbra{D_{ab}}{D_{ab}}\Bigr)~,
\end{split}
\label{fulldickedemo}
\end{equation}
where the numerical factors are due to normalisation.
Eq.(\ref{fulldickedemo}) clearly shows that each $\rho_j^i$ corresponds to a bipartite separable DS state. Nevertheless, this is not sufficient to conclude the proof since we still have to show that, starting from the previous matrices $\rho_{j}^{i}$, it is possible to construct a tripartite separable DS state. 
The next step of the proof consists on realising that, given a matrix $M_i$, this matrix lives in the space defined by the elements of the basis $\{\ket{i00},\ket{i11},\dots,\ket{idd}\}$ and, therefore, the tripartite state that arises from $\rho_{i}$, for a given $\rho_j^i$, can be obtained setting
\begin{equation*}
    \rho_i= \ketbra{i}{i}\otimes\sum_j\rho_j^i~.
\end{equation*}
Notice, however, that this state is separable but not symmetric.
Without loss of generality, as $\rho_j^i$ is separable,
we can cast $\rho_j^i=\sum_{k} \mu_{k} (\rho^{(k)}_A)_{j}^{i}\otimes (\rho^{(k)}_B)_{j}^{i}$. Since we want $\rho_{i}$ to be invariant under the permutation of any of the parties, we must perform an equally weighted convex combination of the form:
\begin{equation}
    \rho_i=\frac{1}{3}\sum_{j,k} \mu_{k}\Bigl( \ketbra{i}{i}\otimes (\rho^{(k)}_A)_{j}^{i}\otimes (\rho^{(k)}_B)_{j}^{i} +(\rho^{(k)}_A)_{j}^{i} \otimes\ketbra{i}{i}\otimes (\rho^{(k)}_B)_{j}^{i}+(\rho^{(k)}_A)_{j}^{i}\otimes (\rho^{(k)}_B)_{j}^{i}\otimes\ketbra{i}{i}\Bigr)~.
\end{equation}

The resulting state $\rho_i$ is indeed DS 
and, by construction, it is separable being a convex combination of separable states. Finally, we must take into account that each $\rho_i$ needs to be normalised. For this reason, we introduce the normalised state $\tilde{\rho_i}$, defined as
\begin{equation}
\label{fin_state}
    \tilde{\rho}_i=\frac{1}{d2^d}\frac{\rho_i}{||M_i(\rho_{DS})||_1}~,
\end{equation}
and, we construct a convex combination with weights equal to $\frac{||M_i(\rho_{DS})||_1}{||M(\rho_{DS})||_1}$, which leads to the desired state
\begin{equation}
    \rho=\frac{1}{d2^d}\frac{1}{||M(\rho_{DS})||_1}\sum_i\rho_i~.
\end{equation}
Since each $\rho_i$ is a tripartite separable DS state, the same holds true for the state $\rho$ in Eq.(\ref{fin_state}), thus concluding the proof.
\end{proof}

\section{Proof of Theorem \ref{mr2}}
\label{p2}

\begin{proof}
The proof of this Theorem consists in showing that the conditions that guarantee that $M(\rho_{DS}) \succeq 0$ are identical to the ones that ensure PPT of the state $\rho_{DS}^{\Gamma_{i}}$. Recalling the observation made in the main text, it is sufficient to prove that the implication holds true for one particular partition, since $M(\rho_{DS})$ is univocally specified for tripartite systems. Then, one can see that the result follows from the fact that $M(\rho_{DS})$ is constructed expressing $\rho_{DS}^{\Gamma}$ as a direct sum of smaller matrices and erasing the repeated rows and columns as well as the 1x1 matrices. Notice that this latter operation does not affect the positivity of a matrix. In the following, we prove the result for $d=3$ although we stress that the proof is analogous for any $d$. For $d=3$, the partial transposition of a DS state $\rho_{DS}$ can be expressed as:
\begin{align}
\rho_{DS}^{\Gamma}=&\begin{pmatrix}
        p_{000}&\frac{p_{001}}{3}&\frac{p_{002}}{3}&\frac{p_{001}}{3}&\frac{p_{002}}{3}\\
        \frac{p_{001}}{3}&\frac{p_{011}}{3}&\frac{p_{012}}{6}&\frac{p_{011}}{3}&\frac{p_{012}}{6}\\
        \frac{p_{002}}{3}&\frac{p_{012}}{6}&\frac{p_{022}}{3}&\frac{p_{012}}{6}&\frac{p_{022}}{3}\\
        \frac{p_{001}}{3}&\frac{p_{011}}{3}&\frac{p_{012}}{6}&\frac{p_{011}}{3}&\frac{p_{012}}{6}\\
        \frac{p_{002}}{3}&\frac{p_{012}}{6}&\frac{p_{022}}{3}&\frac{p_{012}}{6}&\frac{p_{022}}{3}\\
    \end{pmatrix}  \bigoplus \begin{pmatrix}
        \frac{p_{001}}{3}&\frac{p_{001}}{3}&\frac{p_{011}}{3}&\frac{p_{012}}{6}&\frac{p_{012}}{6}\\
        \frac{p_{001}}{3}&\frac{p_{001}}{3}&\frac{p_{011}}{3}&\frac{p_{012}}{6}&\frac{p_{012}}{6}\\
        \frac{p_{011}}{3}&\frac{p_{011}}{3}&p_{111}&\frac{p_{112}}{3}&\frac{p_{112}}{3}\\
        \frac{p_{012}}{6}&\frac{p_{012}}{6}&\frac{p_{112}}{3}&\frac{p_{122}}{3}&\frac{p_{122}}{3}\\
        \frac{p_{012}}{6}&\frac{p_{012}}{6}&\frac{p_{112}}{3}&\frac{p_{122}}{3}&\frac{p_{122}}{3}\\
    \end{pmatrix}   \bigoplus \begin{pmatrix}
        \frac{p_{002}}{3}&\frac{p_{012}}{6}&\frac{p_{002}}{3}&\frac{p_{012}}{6}&\frac{p_{022}}{3}\\
        \frac{p_{012}}{6}&\frac{p_{112}}{3}&\frac{p_{012}}{6}&\frac{p_{112}}{3}&\frac{p_{122}}{3}\\
        \frac{p_{002}}{3}&\frac{p_{012}}{6}&\frac{p_{002}}{3}&\frac{p_{012}}{6}&\frac{p_{022}}{3}\\
        \frac{p_{012}}{6}&\frac{p_{112}}{3}&\frac{p_{012}}{6}&\frac{p_{112}}{3}&\frac{p_{122}}{3}\\
        \frac{p_{022}}{3}&\frac{p_{122}}{3}&\frac{p_{022}}{3}&\frac{p_{122}}{3}&p_{222}\\
    \end{pmatrix} \bigoplus \begin{pmatrix}
       \frac{p_{012}}{6}&\frac{p_{012}}{6}\\
        \frac{p_{012}}{6}&\frac{p_{012}}{6}\\
    \end{pmatrix} \bigoplus_{i\neq j} \frac{p_{iij}}{3}~,
\end{align}
where the $2 \times 2$ matrix has multiplicity 3.
Let us apply Sylvester's criterion, which states that a matrix A is semidefinite positive if and only if all of its principal minors are non-negative. First, notice that every $1\times1$ matrix is non-negative, since the coefficients $p_{ijk}$ represent probabilities. Second, due to their explicit expression, the $2 \times 2$ matrix are easily found to be doubly non-negative. As for the remaining three $5\times5$ matrices, it can be seen that every principal minor of order 4 and 5 will be equal to 0, since there are repeated rows and columns in all of these submatrices. Finally, regarding the $3\times3$ minors, it can be seen that the only submatrices that do not have a 0 determinant are:
\begin{equation}
\begin{split}
    \begin{pmatrix}
        p_{000} & p_{001}/3 & p_{002}/3 \\
        p_{001}/3 & p_{011}/3& p_{012}/6 \\
        p_{002}/3 & p_{012}/6 & p_{022}/3
    \end{pmatrix}~, 
    \begin{pmatrix}
        p_{001}/3 & p_{011}/3 & p_{012}/6 \\
        p_{011}/3 & p_{111} & p_{112}/3 \\
        p_{012}/6 & p_{112}/3 & p_{122}/3
    \end{pmatrix}~,
    \begin{pmatrix}
        p_{002}/3 & p_{012}/6 & p_{022}/3 \\
        p_{012}/6 & p_{112}/3 & p_{122}/3 \\
        p_{022}/3& p_{122}/3 & p_{222}
    \end{pmatrix}~,
\end{split}
\end{equation}
which correspond to the three possible $M_i(\rho_{DS})$ that arise from the general expression for $M_i(\rho_{DS})$ defined 
in the main text in Eq.(\ref{miforma}) for $d=3$. As a consequence, it follows that imposing PPT of the state $\rho_{DS}$ implies $M(\rho_{DS}) \succeq 0$ and vice versa.
For $d>3$, the partial transposition of a DS state $\rho_{DS}$ presents an analogous direct sum structure. In particular, erasing the repeated rows and columns as before leads to $d$ matrices $M^{(d)}_i(\rho_{DS})$ of the form:
\begin{equation}
    M^{(d)}_i(\rho_{DS})=  \begin{pmatrix}
        \bar{p}_{i,0,0} & \bar{p}_{i,0,1} & \ldots&\bar{p}_{i,0,d-1}\\
        \bar{p}_{i,0,1} & \bar{p}_{i,1,1} & \ldots&\bar{p}_{i,1,d-1} \\
        \vdots&\vdots&\ddots&\vdots\\
         \bar{p}_{i,0,d-1} & \bar{p}_{i,1,d-1} & \ldots& \bar{p}_{i,d-1,d-1}~.
    \end{pmatrix}  
\end{equation}
Hence, if $\rho_{DS}$ is PPT then $M^{(d)}_i(\rho_{DS}) \succeq 0$ for every $i$ and vice versa, from which it follows that $M(\rho_{DS}) \in \mathcal{DNN}_{d^2}$, concluding the proof.
\end{proof}

\section{Proof of Theorem \ref{mr1n}}
\label{mr1nproof}
\begin{proof}    
As explained in the main text, here we prove only the $(\Leftarrow)$ implication of theorem \ref{mr1n}, since the other direction is analogous to the tripartite case. The main difference lies in the fact that, recalling Eq.(\ref{def4parties}), $M^{(4)}$ can be cast as $M^{(4)} = \Bar{M}^{(4)} \bigoplus_{i<j} M^{(4)}_{ij} $, where $\Bar{M}^{(4)} $ is a $d(d+1)/2$ matrix and $ M^{(4)}_{ij}$ are $d(d-1)/2$ matrices of size $d$, each of them being completely positive. As for the matrices $M^{(4)}_{ij}$, we follow the same approach of the tripartite case, that is, each of these matrices will lead to a bipartite separable DS state. The only difference with the previous case is that now we have to add two parties to form a four-partite state. As before, we define our four-partite state as
\begin{equation}
    \rho_{ij}=\frac{1}{|\mathcal{G}_4|}\sum_{\pi \in \mathcal{G}_4,k}\mu_k\pi(\ketbra{i}{i}\otimes\ketbra{j}{j}\otimes(\rho_A^{(k)})_{ij}\otimes(\rho_B^{(k)})_{ij})~,
\end{equation}
where $\mathcal{G}_{4}$ is the group of permutations of $4$ elements.
As for the matrix $\Bar{M}^{(4)}$, being completely positive, we can write:
 \begin{equation}
     \Bar{M}^{(4)}=\sum BB^T=\sum_i \mathbf{b}_i\mathbf{b}_i^T
 \end{equation}
Introducing the states
\begin{equation}
    \ket{\xi_{i,\mathbf{k},s}}=\sum_{0\leq l<d} (-1)^{\mathbf{k}_l}w^{sl}\chi_{i,l}\ket{l}, \quad |\chi_{i,l}|^2 =(\mathbf{b}_i)_l
\end{equation}
we can construct the following four-partite state:
\begin{equation}
    \Bar{\rho}_i=\sum_{\substack{0\leq s<d\\0\leq \mathbf{k}<2^d}}\ket{\xi_{i,\mathbf{k},s}}^{\otimes 4}\bra{\xi_{i,\mathbf{k},s}}^{\otimes 4}~.
\end{equation}
Repeating the same steps of the proof of Theorem \ref{mr1}, we end up with two sums over independent indices, $s$ and $\mathbf{k}$. Performing the sum over $s$ yields a factor $ d \delta_{\substack{l_1+l_2+l_3+l_4,\\l_5+l_6+l_7+l_8}}$, while the sum over $\mathbf{k}$ selects the following partitions of indices, i.e., $(8,0,0,0,0,0,0,0)$, $(6,2,0,0,0,0,0,0)$, $(4,4,0,0,0,0,0,0)$, $(4,2,2,0,0,0,0,0)$ and $(2,2,2,2,0,0,0,0)$. Combining these constraints one recovers the following state:
\begin{align}
    \Bar{\rho}=\sum_i\Bar{\rho_i}=\sum_i d2^d\Bigl(&\sum_{0\leq a< d}|\chi_{i,a}|^8\ketbra{D_{aaaa}}{D_{aaaa}}+\sum_{\substack{0\leq a,b <d,\\a\neq b}} 4|\chi_{i,a}|^6|\chi_{i,b}|^2\ketbra{D_{aaab}}{D_{aaab}}+\sum_{\substack{0\leq a,b <d,\\a\neq b}}6|\chi_{i,a}|^4|\chi_{i,b}|^4\ketbra{D_{aabb}}{D_{aabb}}+\\&\sum_{\substack{0\leq a,b,c <d,\\a\neq b\neq c}} 12|\chi_{i,a}|^4|\chi_{i,b}|^2|\chi_{i,c}|^2\ketbra{D_{aabc}}{D_{aabc}}+\sum_{\substack{0\leq a,b <d,\\a\neq b\neq c\neq d}} 24|\chi_{i,a}|^2|\chi_{i,b}|^2|\chi_{i,c}|^2|\chi_{i,d}|^2\ketbra{D_{abcd}}{D_{abcd}}\Bigr)~\nonumber,
\end{align}
which is easily seen to be a four-partite separable DS state. Now, to obtain the final state, we perform a convex combination of all the four-partite states obtained previously in this proof, i.e.,
\begin{equation}
    \rho=\frac{1}{d2^d}\frac{1}{||M(\rho_{DS})||_1}\left(\Bar{\rho}+\sum_{i\neq j}\rho_{ij}\right)~.
\end{equation}
By construction, the resulting state is also separable and DS, thus concluding the proof.
\end{proof}


\end{document}